\documentclass[a4paper]{amsart}

\usepackage[utf8]{inputenc}
\usepackage[T1]{fontenc}
\usepackage{graphicx}
\usepackage{graphicx}
\usepackage{color}
\usepackage{amsbsy}
\usepackage{amssymb}
\usepackage{amsmath}
\usepackage{amsfonts}
\usepackage{dsfont}
\usepackage{enumerate}
\usepackage{eurosym}
\usepackage{todonotes}
\usepackage{cite} 
\usepackage{comment}
\usepackage{mathtools}
\usepackage{stix}

\newtheorem{theorem}{Theorem}
\newtheorem{corollary}[theorem]{Corollary}
\newtheorem{proposition}[theorem]{Proposition}
\newtheorem{lemma}[theorem]{Lemma}

\newtheorem{algorithm}[theorem]{Algorithm}

\def\C{{\mathds C}}
\def\E{{\mathds E}}

\def\1{{\mathds 1}}
\def\M{{\mathds M}}

\def\P{{\mathds{P}}}

\def\V{{\mathds V}\!\!\operatorname{ar}}
\def\R{{\mathds{R}}}

\def\cE{{\mathcal E}}

\def\cL{{\mathcal L}}

\def\cO{{\mathcal O}}
\def\co{o}

\def\d{\delta}
\def\e{\varepsilon}

\def\k{{\kappa}}
\def\l{\lambda}

\def\s{\sigma}

\def\Th{{\Theta}}

\def\tcr{{\rm cr}}
\def\rcr{{\overline{\rm cr}}}
\def\st{{\rm stress}}

\def\cov {{\C \rm{ov}}}
\def\cor {{\C \rm{orr}}}

\def\vol{\l}
\def\surf{\mathrm{surf}}

\title{Crossing Numbers and Stress of Random Graphs}

\author[M. Chimani]{Markus Chimani}
\author[H. Döring]{Hanna Döring}
\author[M. Reitzner]{Matthias Reitzner}

\address{Theoretical CS, Institute of Computer Science, Uni Osnabrück, Germany}
\email{markus.chimani@uni-osnabrueck.de} 
\urladdr{https://orcid.org/0000-0002-4681-5550}
\address{Stochastics, Institute of Mathematics, Uni Osnabrück, Germany}
\email{hanna.doering@uni-osnabrueck.de}
\address{Stochastics, Institute of Mathematics, Uni Osnabrück, Germany}
\email{matthias.reitzner@uni-osnabrueck.de}

\thanks{A preliminary version of this research appeared in Proc.\ Int'l Symp.\ on Graph Drawing and Network Vis.\ (GD 2018), LNCS 11282, pp.\ 255--268.}

\renewcommand{\paragraph}[1]{\subsubsection*{#1}}

\begin{document}

\begin{abstract}
Consider a random geometric graph over a random point process in $\R^d$. Two points are connected 
by an edge if and only if their distance is bounded by a prescribed distance parameter. We show that 
projecting the graph onto a two dimensional plane is expected to yield a constant-factor crossing number (and rectilinear crossing number) 
approximation. We also show that the crossing number is positively correlated to the stress of the graph's projection.

\bigskip

\end{abstract}

\maketitle

\section{Introduction}\label{sec:intro}
An undirected \emph{abstract} graph $G_0$ consists of vertices and edges connecting vertex pairs.
An \emph{injection} of $G_0$ into $\R^d$ is an injective map from the vertices of $G_0$ to $\R^d$,
and edges onto curves between their corresponding end points but not containing any other vertex point. 
For $d\geq 3$, we may assume that distinct edges do not share any point (other than a common end point).
For $d=2$, we call the injection a \emph{drawing}, and it may be necessary to have points where curves \emph{cross}. A drawing is \emph{good} 
if no pair of edges crosses more than once, nor meets tangentially, and no three edges share the same crossing point.
Given a drawing $D$, we define its crossing number $\tcr(D)$ as the number of points where edges cross. 
The crossing number $\tcr(G_0)$ of the graph itself is the smallest $\tcr(D)$ over all its good drawings~$D$.
We may restrict our attention to the \emph{rectilinear crossing number} $\rcr(G_0)$, where edge curves are straight lines;
note that $\rcr(G_0)\geq \tcr(G_0)$.

The crossing number and its variants have been studied for several decades, see, e.g.,~\cite{S-book}, but
still many questions are widely open. We know the crossing numbers only for very few graph classes; 
already for $\tcr(K_n)$, i.e., for complete graphs with $n$ vertices,
we only have conjectures, and for $\rcr(K_n)$ not even them.
Since deciding $\tcr(G_0)$ is NP-complete~\cite{GJ83} (and $\rcr$ even $\exists\R$-complete~\cite{B91}), several attempts for approximation algorithms
have been undertaken. The problem does not allow a PTAS unless $\mathrm{P}{=}\mathrm{NP}$~\cite{C13}. 
For general graphs, we currently do
not know whether there is an $\alpha$-approximation for any constant~$\alpha$ . However, we can achieve constant ratios for dense graphs~\cite{cit:fox-pach-suk} and for bounded pathwidth graphs~\cite{BCDM17}.
Other strong algorithms deal with graphs of maximum bounded degree and
achieve either slightly sublinear ratios~\cite{Chuzhoy}, or constant ratios for further restrictions such as embeddability on low-genus surfaces~\cite{ghls,torusHS,surfaceApprox}
or a bounded number of graph elements to remove to obtain planarity~\cite{CM11,apex,TighterInsertion,exactmei}.

We will make use of the \emph{crossing lemma}, originally due to~\cite{ACNS82,Leighton}%
\footnote{Incidentally, the lemma allows an intriguingly elegant proof using stochastics~\cite{thebook}.}:
There are constants%
\footnote{The currently best constants $d=7,c=\frac1{20}$ are due to~\cite{deKlerk}.} 
$d\geq 4,c\geq \frac1{64}$ such that any abstract graph $G_0$ on $n$ vertices and $m\geq dn$ edges has $\tcr(G_0)\geq c m^3/n^2$. In particular
for (dense) graphs with $m=\Theta(n^2)$, this yields the asymptotically tight maximum of $\Theta(m^2)$ crossings.

\paragraph{Random Geometric Graphs (RGGs).}
We always consider a \emph{geometric} graph $G$ as input, i.e., an \emph{abstract} graph $G_0$ together with
a straight-line injection into $\R^d$, for some $d\geq 2$; we identify the vertices with their points. 
For a two-dimensional plane $L$, the postfix operator~$|_L$ denotes the projection onto $L$.

Given a set of points $V$ in $\R^d$, the \emph{unit-ball graph} (\emph{unit-disk graph} if $d=2$) 
is the geometric graph using $V$ as vertices that has an edge between two points iff
balls of radius~1 centered at these points touch or overlap. Thus, points are adjacent iff their distance is~$\leq 2$.
In general, we may use arbitrary threshold distances $\delta>0$.
We are interested in \emph{random} geometric graphs \emph{(RGGs)}, i.e., when using a Poisson point process to obtain $V$ for the above graph class.

\paragraph{Stress.}
When drawing (in particular large) graphs with straight lines in practice, \emph{stress} is a well-known and successful concept, see, e.g.,~\cite{KobHandbook,klimentaBrandes,stresscompare}: 
let $G$ be a geometric graph, $d_0,d_1$ two distance functions on vertex pairs---(at least) the latter of which depends on an injection---and $w$ weights. We have:
\begin{equation}\label{def:stress}
\st(G) := \sum_{v_1, v_2\in V(G), v_1\neq v_2} w(v_1,v_2) \cdot ( d_0(v_1,v_2) - d_1(v_1,v_2) ) ^2.
\end{equation}
In a typical scenario, $G$ is injected into $\R^2$, $d_0$ encodes the graph-theoretic distances (number of edges on the shortest path) 
or some given similarity matrix, and $d_1$ are the Euclidean distances in $\R^2$.
Intuitively, in a drawing of 0 (or low) stress, the vertices' geometric distances $d_1$ are (nearly) identical to their ``desired'' distance according to $d_0$.
A typical weight function $w(v_1,v_2):=d_0(v_1,v_2)^{-2}$ softens the effect of ``bad'' geometric injections for vertices that are far away from each other anyhow.
It has been observed \emph{empirically} that low-stress drawings \emph{tend} to be visually 
pleasing and to have a low number of crossings, see, e.g.,~\cite{crstress,kobourov}. While it may seem worthwhile to approximate the crossing number by minimizing a drawing's stress, 
there is no sound mathematical basis for this approach.

There are different ways to find (close to) minimal-stress drawings in 2D
~\cite{stresscompare}. 
One way is multidimensional scaling, cf.~\cite{klimentaBrandes}, where we start with an injection of an abstract graph $G_0$ into some high-dimensional
space $\R^d$ and asking for a projection of it onto $\R^2$ with minimal stress. 
It should be understood that Euclidean distances in a unit-ball graph in $\R^d$
by construction closely correspond to the graph-theoretic distances. In fact, for such graphs it seems reasonable to use the distances in $\R^d$ as the given metrics $d_0$, 
and seek an injection into $\R^2$---whose resulting distances form $d_1$---by means of projection.

\paragraph{Contribution.}
We consider RGGs for large $t$ 
and
investigate the mean, variance, and corresponding law of large numbers both for their rectilinear crossing number and their minimal stress when projecting them onto the plane. 
We also prove, for the first time, a positive correlation between these two measures.

While our technical proofs make heavy use of stochastic machinery, 
the consequences are very algorithmic: We give a surprisingly simple
algorithm that yields an \emph{expected constant} approximation ratio for random geometric graphs
even in the pure abstract setting. In fact, we can state
the algorithm already now; the remainder of this paper deals with the proof of its properties and correctness:

\begin{algorithm}
Given a random geometric graph $G$ in $\R^d$,
we pick a \emph{random} two-dimensional plane $L$ in $\R^d$ to obtain a straight-line drawing $G|_L$ by projecting $G$ onto $L$.
\end{algorithm}

Throughout this paper, we prefer to work within the setting of a Poisson point process 
because of the strong mathematical tools that are available in this case. 
It is straightforward to transfer the computations to the binomial setting (which is the case of $n$ uniform random points in $W$) instead of a Poisson point process. As a particular example we compute in Section \ref{sec:dePoissonize} the expected crossing number in this setting.

\section{A Stroll through the Core Results}

In this paper, we combine questions of graph theory with modern approaches from stochastics. As such, we are aware that many
readers typically interested in crossing numbers and related concepts, may not be intrinsically versed in the finer details of the latter. 
Yet, they may want to understand the graph theoretic and algorithmic consequences.
As such, before we go into the details of actually proving our results rigorously, we summarize our findings in this section in a hopefully
comparably light-reading style. 
Some readers, when only interested in the gist of the results and their algorithmic and graph theoretic 
consequences, may find this section sufficient; others may deem it a helpful start to understand the organization of the subsequent sections.
Generally, the results stated within this section will often only describe bounds in terms of asymptotics (using $\mathcal O$-notation); more specific constants
can be found in the detailed results in the sections thereafter.

We always consider \emph{random} geometric graphs that sample the vertex set by a random process. A traditional
approach, in particular in algorithmics, would be to first decide on a number $n$ of vertices, and then sample each point $p$ in $\R^d$ by picking
each of the $d$ coordinates independently and uniformly at random between some upper and lower bounds. For uniform bounds, this sample space is a $d$-dimensional
hypercube, and as such \emph{not rotation invariant}: intuitively, this means that when projecting the high-dimensional point set onto a 2D-plane, 
different angular positions of the plane relative to the coordinate system will, in general, lead to vastly different results. In contrast to this,
we will sometimes also consider a mathematically more well-behaved \emph{rotation invariant} sample spaces: the one where each point is 
picked uniformly at random within a $d$-dimensional hypersphere. We note that algorithmically sampling this space, although not as trivial as the aforementioned process,
is well understood. 

Furthermore, our results are mathematically easier to show for a \emph{Poisson point process} instead of for a predetermined number of vertices $n$.
Thereby, the resulting number of points itself is subject to the random process, and we only prescribe the \emph{expected} number of vertices. 
The intuitive benefit is that at any point in time, we have a non-zero probability of a further point being sampled. 
This will probably become most evident when discussing the computation of variances and covariances, where the so-called \emph{Slivnyak-Mecke formula} plays a crucial role. 
However, see Section~\ref{sec:dePoissonize} for a discussion on the results for ``binomial input''. 

Throughout the paper, all our upper bounds arise from projections of a (potentially high-dimensional) straight-line drawing; they are thus witnesses for $\rcr$.
All our lower bounds are based on the crossing lemma for general drawings, witnessing $\tcr$.
Thus, since $\tcr(G_0)\leq \rcr(G_0)$ for any abstract graph $G_0$, all our results hold for $\tcr$ and $\rcr$ simultaneously.

The first important tool to be used is that of \emph{U-statistics}, to be defined in the subsequent Section~\ref{sec:tools}. Intuitively, it describes properties of measures
that allow for certain strong stochastic machinery to be used. 
We will see that the rectilinear crossing number and the stress of a graph are such U-statistics.
Based thereon, we are then able to derive in Section~\ref{sec:Ecr}:

\begin{theorem}[see Theorem~\ref{th:exp-rcr}]
	Let $G_0$ be the underlying abstract graph of a random geometric graph $G$ in $\R^d$ with $n$ vertices and $m$ edges. 
	For \emph{any} two-dimensional projection plane $L$, we have:
	$$\tcr(G_0)\;\leq\; \rcr(G_0) \;\leq\; \E_V \rcr(G|_L) \;\in\; \Theta\left(\frac{m^3}{n^2}\cdot\left(\frac{m}{n^2}\right)^{\frac{2-d}{d}}\right).$$
\end{theorem}

Observe that the fraction $\frac{m}{n^2}$ is essentially the density of the graph, measured as its ratio w.r.t.\ a complete graph.

Once the expectation is established, it is interesting to understand how close actual randomized computations are
expected to be to this expected value.
In Section~\ref{sec:var-rcr}, Theorem~\ref{th:var-rcr} we will show that the variance $\V_{V}\rcr(G|_L)\ll \E_V\rcr(G|_L)$ and a corresponding
law of large numbers (Corollary~\ref{cor:llnrcr}). This allows us to deduce: 
\begin{corollary}
	The observed crossing number will be very close to $\E_V\rcr(G|_L)$ with high probability.
\end{corollary}	

Combining these two results with the above crossing lemma $\rcr(G_0)\geq \tcr(G_0)\in \Omega({m^3}/{n^2})$, we directly obtain our two central approximation results. They
constitute the first (expected) crossing number approximations for a rich class of randomized graphs:
\begin{corollary}
 Let $G$ be a random geometric graph in $\R^2$ (unit-disk graph), and $G_0$ its underlying abstract graph. \emph{With high probability}, the number of crossings in its natural straight-line drawing 
 is at most a \emph{constant factor} away from $\tcr(G_0)$ and $\rcr(G_0)$.
\end{corollary}
\begin{corollary}
 Let $G$ be a random geometric graph in $\R^d$ (unit-ball graph, for some constant $d\geq 3$), and $G_0$ its underlying abstract graph. We obtain a straight-line drawing $D$ by projecting it onto a randomly chosen two-dimensional plane.
 \emph{With high probability}, the number of crossings in $D$ is at most a \emph{factor} $\alpha$ away from $\tcr(G_0)$ and $\rcr(G_0)$. Thereby, $\alpha$ is only dependent on the graph's density.
 If the latter is constant, so is $\alpha$.
\end{corollary}

In Section~\ref{sec:isotropic-rcr} we turn our attention to asking whether it may be beneficial to not pick some arbitrary fixed projection plane, but to choose the plane randomly.
This leads to the question whether we can find a particularly good projection plane by trying several ones at random.
For this question, we restrict ourselves to rotation invariant sampling spaces. Considering non-invariant spaces does not seem futile, but would come at the cost
of very tedious and lengthy computations, without us learning much more than from the invariant case. The obvious benefit of the latter is that
$\V_L \E_V \rcr(G|_L)=0$, i.e., all planes have the same chance of being ``good''.

In Theorem~\ref{th:isotropic-rcr} we show that, when picking $L$ in the invariant setting, both the expected rectilinear crossing number and its variance are of the same
asymptotic order as for the previous case, but now with better constants. We again yield a corresponding law of large numbers (Corollary~\ref{cor:llninv}). However, the most interesting result here is:
\begin{theorem}[see Theorem \ref{thm:manytries}]
	The chance of picking a plane that yields a significantly better crossing number for an $n$-vertex RGG than the expected value is in $\cO(1/\exp({\sqrt[4]t}))$ where $\exp$ is an exponentially growing function.
\end{theorem}
Thus, we are expected to require an exponential number of tries in order to find a particularly ``good'' projection plane, i.e., one that yields significantly fewer crossings than any random plane.
This complexity results suggests that it may be worthwhile to combine the search for an optimal projection plane with other techniques.

To this end, we turn our attention to the stress of RRGs: as discussed in Section~\ref{sec:intro}, one way of obtaining stress-minimal drawings is to use multi-dimensional scaling---a close relative
to our projection problem.
In Section~\ref{sec:stress} we show expectation and variance for stress (Theorem~\ref{th:stress}) using analogous techniques as above. 
In this context, however, we do \emph{not} consider the graph theoretic distance as our target metric $d_0$, but the distances in the high-dimensional injection intrinsic to $G$.
We then use these results in Section~\ref{sec:correlation} to, for the first time, prove a strictly positive correlation between the optimal stress of a graph and its crossing number:

\begin{theorem}[see Theorem \ref{thm:cov} and Corollary \ref{cor:corr-VL-isotropic}]
Let $G$ be a random geometric graph $G$, not neccessarily in the rotation invariant case, and $L$ a two-dimensional plane that may or may not be chosen at random.
Then the correlation between the expected rectilinear crossing number and the expected stress is positive (as the number of vertices tends to infinity).
\end{theorem}

Loosely speaking, this tells us that optimizing the stress of a graph drawing, we are---in expectation---at the same time automatically optimizing for the crossing number. 

We need to be careful in understanding the above scenarios, however: Our results always assume a random input graph. There are, in fact, specific graphs that may arise in 
the random process (sporadically), for which there is neither a reasonable projection (w.r.t.\ its crossing number), nor do they exhibit a positive correlation between 
crossing number and stress.

\section{Notations and Tools from Stochastic Geometry}\label{sec:tools}
Let $W \subset \R^d$ be a convex set of volume $\vol_d(W)=1$. 
Choose a Poisson distributed random 
variable $n$ with parameter $t$, i.e., $\E n =t$.
Next choose $n$ points $V=\{v_1, \dots, v_n\}$ independently in $W$ according to the uniform 
distribution. 
Those points form a Poisson point process $V$ in $W$ of intensity~$t$.
A Poisson point process has several nice properties, e.g., for disjoint subsets $A,B\subset W$, 
the sets $V\cap A$ and $V\cap B$ are independent (thus also their size is 
independent). 
Let $V^k_{\neq}$, $k\geq 1$, be the set of all ordered $k$-tuples over $V$ with pairwise distinct elements.
We will consider $V$ as the vertex set of a geometric graph $G$
for the distances parameter $(\d_t)_{t>0}$ with edges $E=\{ \{u,v\} \mid u,v\in V, u\neq v,
\|u-v\| \leq \delta_t \}$, i.e.,
we have an edge between two distinct points if
and only if their distance is at most $\d_t$.
Such \emph{random geometric graphs (RGG)} have been extensively investigated, 
see, e.g., \cite{Penrose03, CTMRMS}, but nothing is known
about the stress or crossing number of its underlying abstract graph~$G_0$.

A \emph{U-statistic} $U_k(f):=\sum_{\mathbf{v}\in V^k_{\neq}} f(\mathbf{v})$ 
is the sum over $f(\mathbf{v})$ for all
$k$-tuples~$\mathbf{v}$. Here, $f$ is a measurable non-negative real-valued function, and 
$f(\mathbf{v})$ only depends on
$\mathbf{v}$ and is independent of the rest of $V$.
The number of edges in $G$ is a U-statistic as 
$m=\frac12 \sum_{v,u\in V, v\neq u} \1(\| v-u\| \leq \d_t)$. Likewise, the stress of a geometric graph as well as the crossing number of a 
straight-line drawing is a U-statistic, using 2- and 4-tuples of $V$, respectively.
The well-known multivariate Slivnyak-Mecke formula, which we use here in a simple form, tells us how to compute the
expectation $\E_V$ over all realizations of the Poisson process $V$; for U-statistics we have, 
see~\cite[Cor.~3.2.3]{SW}:
\begin{equation}\label{eq:Mecke}
\E_{V} \sum_{(v_1,\ldots,v_k)\in V^k_{\neq}}f(v_1,\ldots,v_k) 
= 
t^k \int\limits_{W^k} 
f(v_1,\ldots,v_k)\,dv_1\cdots dv_k.
\end{equation}
We already know $ \E_V n= \E_V | V| = t$. Solving the above formula for the expected number of 
edges, we obtain
\begin{equation}\label{eq:edges}
\E_V m= \E_V | E | = \frac {\k_d}{2}\, t^2 \d_t^{d} + \cO(t^2 \d_t^{d+1}\, \surf(W)),
\end{equation}
where $\k_d=\vol_d(B_d)$ is the volume of the unit ball $B_d$ in $\R^d$, and $\surf(W)$ the surface 
area of $W$.
For $n$ and $m$, central limit theorems and concentration inequalities are well known as $t \to 
\infty$, see, e.g., 
\cite{Penrose03,CTMRMS}. 

The expected degree $\E_V \deg(v)$ of a typical vertex $v$ is
approximately of order $\kappa_d\, t\, \d_t^d$ (this can be made precise using Palm 
distributions). This naturally leads to three different asymptotic \emph{regimes} as introduced in 
Penrose's book~\cite{Penrose03}:
\begin{itemize}
	\item in the \emph{sparse regime} we have
	$\lim_{t\to\infty}t \, \d_t^d=0$, thus $\E_V \deg(v)$ tends to zero;
	\item in the \emph{thermodynamic regime} we have
	$\lim_{t\to\infty}t \, \d_t^d=c >0 $, thus $\E_V \deg(v)$ is asymptotically constant;
	\item in the \emph{dense regime} we have
	$ \lim_{t\to\infty}t \, \d_t^d=\infty $,
	thus $\E_V \deg(v)\to\infty$.
\end{itemize}
Observe that in standard graph theoretic terms, the \emph{thermodynamic regime} leads to 
\emph{sparse graphs}, i.e., via~\eqref{eq:edges} we obtain
$\E_V m 
= \Th(t) = \Th(\E_V n)$. 
Similarly, the dense regime---together with $\d_t \to c$---leads to \emph{dense graphs}, i.e., 
$\E_V m 
= \Th(t^2) = \Th((\E_V n)^2)$.
Recall that to employ the crossing lemma, we want $m \geq 4n$. 
Also, the lemma already shows that any good (straight-line) drawing of a dense graph $G_0$ already 
gives a constant-factor approximation for $\tcr(G_0)$ (and $\rcr(G_0)$). 
In the following we thus assume a constant $0<c \leq t\,\d_t^{d}$ and $\d_t\to 0$, i.e., $m=\co(n^2)$.

The Slivnyak-Mecke formula is a classical tool to compute expectations 
and will thus be used extensively throughout this paper. 
Yet, suitable methods to compute variances came up only recently. They emerged in 
connection with the development of the Malliavin calculus for Poisson point processes~\cite{LastPenrose11, PR16}. This is due to the fact that the Wiener-It\^{o} chaos expansion is particularly well-behaved for Poisson U-statistics~\cite{MRMS}.
An important operator for functions $g(V)$ of Poisson point processes is the 
\emph{difference} (also called \emph{add-one-cost}) operator,
$$ D_v g(V) := g(V \cup \{v\})-g(V),$$
which considers the change in the function value when adding a single further point~$v$.

\section{Rectilinear Crossing Number of an RGG}
Let $\cL$ be the set of all two-dimensional linear planes and 
$L \in \cL$ be a random plane chosen according to a (uniform) Haar 
probability measure on $\cL$.
The drawing  $G_L:= G|_L$ is the projection of $G$ onto $L$.
Let $[u,v]$ denote the segment between vertex points $u,v\in V$. If their distance 
has to be at most $\delta_t$, this condition is explicitly given by an indicator function or the 
domain of the integral.
The rectilinear crossing number of $G_L$ is a U-statistic of order~$4$:
\begin{equation*}
 \rcr (G_L) =
\frac 18 \! \sum_{(v_1, v_2, v_3, v_4) \in V^4_{\neq}} \!\!\! \1([v_1,v_2]|_L\cap [v_3,v_4]|_L \neq 
\emptyset,\, \| v_1- v_2\| \leq \delta_t,\| v_3-v_4\| \leq \delta_t ).
\end{equation*}
Keep in mind that even for the best possible projection we only 
obtain
$\min_{L\in\cL} \rcr(G_L) \geq \rcr(G_0)$.
Analyzing $\E_{V} \min_{L\in\cL} \rcr(G_L) $ is more complicated than $\E_{L,V} \rcr(G_L)$; 
fortunately, we
will not require it.

\subsection{Calculations of some Geometric Integrals} 
The following integrals show up in several 
calculations and play an essential role in determining the asymptotic behavior of the moments of 
$\rcr(G_L)$.
\begin{lemma}\label{le:J1K}
Let $K$ be a convex body and $L$ a 2-dimensional plane in $\R^d$. Define
$$
J^{(1)}_L (K)
:=
\int\limits_{ (\d_t B_d ) \times K \times
(\d_t B_d)} 
\1([0,x ]\vert_L \cap (y+[0,z])\vert_L \neq \emptyset)
\, dz\, dy\, dx
. $$
Then, \quad
$ J^{(1)}_L (K) =
c_d \, \d_t^{2d+2} \vol_{d-2} (L^\perp \cap K) \ (1+\co(1)) 
$ \quad
as $\d_t \to 0$, with 
\begin{equation}\label{def:cd1}
c_d =
8 \pi \kappa_{d-2}^2  \, {\bf B}\!\left( 3,\frac d2 \right)^2
\end{equation}
for ${\bf B}$ being the beta function.
Furthermore, we have 
$ J^{(1)}_L (K)  \leq c_d \d_t^{2d+2} M_{d-2} (K)  $ 
where $M_{d-2}(K)$ is the volume of the maximal $(d-2)$-dimensional section of $K$. 
\end{lemma}
\begin{proof}
We write $y=(y^L, y^{L^\perp})$ with $y^L \in L$, $y^{L^\perp} \in L^\perp$. 
Clearly, for $[y,z]$ to meet $[0,x]$ we need that $y$ is at least 
contained in a cylinder of radius $2 \d_t$ above the origin, $y^L \in 2 \d_t B_2 \subset L$. Using 
Fubini's theorem we obtain
\begin{align*}
J&^{(1)}_L (K)  
=
\int\limits_{ (\d_t B_d ) \times K \times
(\d_t B_d)} 
\1([0,x ]\vert_L \cap (y+ [0,z])\vert_L \neq \emptyset)
\, dz\, dy\, dx
\\
&=
\int\limits_{\mathclap{\substack{(\d_t B_d ) \times\\ (2 \d_t B_2 \cap K\vert_ L)  \times\\
(\d_t B_d)}}} 
\1([0,x\vert_L ] \cap (y^L{+}[0,z\vert_L]) \neq \emptyset) 
\int\limits_{ L^\perp } \1((y^L, y^{L^\perp})\in K) \, dy^{L^\perp} \ 
\, dz\, dy^L\, dx
\\
&=
\int\limits_{\mathclap{\substack{ (\d_t B_d ) \times\\ (2 \d_t B_2 \cap K\vert_ L)  \times\\
(\d_t B_d)}}} 
\1([0,x\vert_L ] \cap (y^L{+}[0,z\vert_L]) \neq \emptyset)  
\cdot \vol_{d-2} ((y^L{+}L^\perp ) \cap K)\ 
dz\, dy^L\, dx
\\ 
&\leq
\d_t^{2d+2} 
\int\limits_{\mathclap{\substack{ B_d \times\\ (2 B_2 )  \times\\ B_d} }}
\1([0,x\vert_L ] \cap (y^L{+}[0,z\vert_L]) \neq \emptyset)  
\cdot \vol_{d-2} ((\d_t y^L{+}L^\perp ) \cap K)\ 
dz\, dy^L\, dx
\\ 
&\leq
c_d \d_t^{2d+2} \max_{u \in \d_t B_d } \vol_{d-2} ((u +L^\perp ) \cap K)
\end{align*}
with 
\begin{align*}
c&_d 
=
\int\limits_{ B_d \times (2 B_2 )  \times B_d} 
\1([0,x\vert_L ] \cap (y^L + [0,z\vert_L]) \neq \emptyset) 
dz\, dy^L\, dx 
\\ &=
\int\limits_{2 B_2 } 
\int\limits_{ B_{d}^2} 
\1([0,x^L ] \cap (y^L + [0,z^L]) \neq \emptyset) 
dz^L\,dz^{L^\perp}\,  dx^L \,dx^{L^\perp} \, dy^L 
\\ &=
\int\limits_{ B_{2}^2} 
\int\limits_{2 B_2 } 
\1(y^L \in [0,x^L ] + [0,-z^L]) \, dy^L 
\kappa_{d-2}^2 ( 1- \|z^L\|)^{\frac {d-2}2} (1- \|x^L\|)^{\frac {d-2}2}
dz^L\,  dx^L  
\\ &=
\kappa_{d-2}^2  \int\limits_{ B_{2}^2} 
V_2( [0,x^L ] + [0,-z^L]) 
(1- \|z^L\|^{\frac {d-2}2} (1- \|x^L\|)^{\frac {d-2}2}
dz^L\,  dx^L .
\end{align*}
Changing to polar coordinates gives
\begin{align*}
c_d &=
\kappa_{d-2}^2  \int\limits_{ S_1^2}\int\limits_{ [0,1]^2} 
|{\rm det} (r_1 u_1 , r_2 u_2) |
(1- r_1)^{\frac {d-2}2} (1- r_2)^{\frac {d-2}2}
r_1 r_2 \, dr_1 \, dr_2 \, du_1 \, du_2  
\\ &=
2 \pi \kappa_{d-2}^2  \, {\bf B}\!\left( 3,\frac d2 \right)^2
\int\limits_{ S_1}
|{\rm det} ( e_1 , u_2)| \, du_2  
\\ &=
2 \pi \kappa_{d-2}^2  \, {\bf B}\!\left( 3,\frac d2 \right)^2
\int\limits_0^{2 \pi}
|\sin \alpha| \, d\alpha   
=
8 \pi \kappa_{d-2}^2  \, {\bf B}\!\left( 3,\frac d2 \right)^2,
\end{align*}
with ${\bf B}(x,y)=\int_0^1 r^{x-1} (1-r)^{y-1} dr$ being the beta function.
%
Analogously, 
$$
J^{(1)}_L (K)
\geq
c_d \d_t^{2d+2}  \1( 2 \d_t B_2 \subset K\vert_L)
\min_{u \in \d_t B_d } \vol_{d-2} ((u +L^\perp ) \cap K).
$$ 
This finishes the proof of the lemma.
\end{proof}

By $K_{-\d}$ we denote the inner parallel set $\{x\colon (x +\d B_d) \subset K\}$ of a 
convex set 
$K$.
\begin{proposition}\label{prop:Int-rcr}
Let $v$ be a point in $W\subset \R^d$. Then for
\begin{equation}\label{defI_W(v)}
I^{(1)}_{W,L}(v)
=
\int\limits_{W^3} \1([v,v_2]\vert_L \cap [v_3,v_4]\vert_L \neq \emptyset,
\, \| v- v_2\| \leq \delta_t,\| v_3-v_4\| \leq \delta_t) \, dv_2 dv_3 dv_4
\end{equation}
it holds that
\begin{equation}\label{eq:lim-IW}
\lim_{\d_t \to 0} \ \frac{I^{(1)}_{W,L}(v)}{\d_t^{2d+2} } 
= c_d 
\vol_{d-2} ((v +L^\perp ) \cap W) 
\end{equation}
as $\d_t \to 0$, with $c_d$ given in \eqref{def:cd1}.
Further we have \quad
$ I^{(1)}_{W,L}(v) \leq c_d \d_t^{2d+2} M_{d-2} (W)  $ \quad
where $M_{d-2}(W)$ is the volume of the maximal $(d-2)$-dimensional section of $W$.  
\end{proposition}

\begin{proof}
We substitute $v_2 = v + x $, $v_3 = v+ y$ and $v_4 = v+ y+z $ and obtain
\begin{eqnarray*}
I^{(1)}_{W,L}(v)
&=& 
\! \int\limits_{(W-v)^2 }  \int\limits_{W-v-y} \!
\1([0,x ]\vert_L \cap [y,y+z]\vert_L \neq \emptyset,
\| x \| \leq \d_t, \| z \| \leq \d_t) \, dz dy dx
\\ &=& 
\int\limits_{ ((W-v)\cap \d_t B_d ) \times (W-v) \times
((W-v-y)\cap \d_t B_d)} \hspace{-1,5cm}
\1([0,x ]\vert_L \cap [y,y+z]\vert_L \neq \emptyset)
\, dz\, dy\, dx
\\ &\leq & 
\int\limits_{ (\d_t B_d ) \times (W-v) \times (\d_t B_d)} 
\1([0,x ]\vert_L \cap [y,y+z]\vert_L \neq \emptyset)
\, dz\, dy\, dx
\\ &= & 
J_L^{(1)} (W-v) 
\end{eqnarray*}
and on the other hand
\begin{eqnarray*}
I^{(1)}_{W,L}(v)
&\geq & 
\int\limits_{ (\d_t B_d ) \times (W-v) \times
(\d_t B_d)} 
\1([0,x ]\vert_L \cap [y,y+z]\vert_L \neq \emptyset)
\\ &&
\hskip2.2cm  \1(\d_t B_d \subset (W-v)) \1(\d_t B_d \subset (W-v-y))
\, dz\, dy\, dx
\\[2ex] &=&
\1(v \subset W_{- \d_t}) 
\int\limits_{ (\d_t B_d ) \times (W_{- \d_t}-v) \times
(\d_t B_d)}\hspace{-1cm} 
\1([0,x ]\vert_L \cap [y,y+z]\vert_L \neq \emptyset)
\, dz\, dy\, dx 
\\ &=&
J_L^{(1)} (W_{-\d_t}-v)\, .
\end{eqnarray*}
Using Lemma \ref{le:J1K} this leads to
\begin{align*}
\lim_{\d_t \to 0} \frac{J_L^{(1)} (W-v)}{\d_t^{2d+2} }
= \lim_{\d_t \to 0} \frac{J_L^{(1)} (W_{-\d_t}-v)}{\d_t^{2d+2} }
&=
\lim_{\d_t \to 0} \frac{I^{(1)}_{W,L}(v)}{\d_t^{2d+2} }
\\ &=
c_d \vol_{d-2} (L^\perp \cap (W-v)).
\end{align*}
\end{proof}

Analogously we prepare the following integral which will show up in the calculation of the variance.
\begin{lemma}\label{le:J2K}
Let $K$ be a convex body and $L$ a 2-dimensional plane in $\R^d$. Define
\begin{multline*}
J^{(2)}_L (K)
:=
\int\limits_{ (\d_t B_d)  \times K^2 \times (\d_t B_d)^2} 
\1([0,x ]\vert_L \cap (y_1+[0,z_1])\vert_L \neq \emptyset)  \cdot 
\\
 \1([0,x ]\vert_L \cap (y_2+[0,z_2])\vert_L \neq \emptyset)
\, dz_1\, dz_2\, \, dy_1\, dy_2\, dx
.
\end{multline*}
Then, \quad 
$ J^{(2)}_L (K) =
c'_d \, \d_t^{3d+4} \vol_{d-2} (L^\perp \cap K)^2 \ (1+\co(1)) 
$ \quad
as $\d_t \to 0$, with 
\begin{eqnarray}\label{def:cd2}
c'_d 
&=&
\pi \kappa_{d-2}^3  \, {\bf B}\!\left( 3,\frac d2 \right)^2  \, {\bf B}\!\left( 4,\frac d2 \right) .
\end{eqnarray}
Furthermore, we have \quad
$ J^{(2)}_L (K)  \leq c'_d \d_t^{3d+4} M_{d-2} (K)^2  .$
\end{lemma}
\begin{proof}(similar to above)
Writing $y_i=(y_i^L, y_i^{L^\perp})$ with $y_1 \in L$, $y_2 \in L^\perp$ 
we obtain
\begin{align*}
J&^{(2)}_L (K)  
= \nonumber
  \int\limits_{ \d_t B_d  \times K^2 \times (\d_t B_d)^2} 
      \1([0,x ]\vert_L \cap (y_1+[0,z_1])\vert_L \neq \emptyset)\\ 
&\hskip3cm
      \1([0,x ]\vert_L \cap (y_2+[0,z_2])\vert_L \neq \emptyset)
      \, dz_1\, dy_1\, \, dz_2\, dy_2\, dx_1
\\[1ex] &=
  \int\limits_{ (\d_t B_d ) \times (2 \d_t B_2 \cap K\vert_ L)^2  \times (\d_t B_d)^2} 
      \1([0,x\vert_L ] \cap (y_1^L + [0,z_1\vert_L]) \neq \emptyset) 
\\ & \hskip3cm 
      \vol_{d-2} ((y_1^L+L^\perp ) \cap K)\ \1([0,x\vert_L ] \cap (y_2^L + [0,z_2\vert_L]) \neq 
\emptyset) 
\\ & \hskip3cm 
      \vol_{d-2} ((y_2^L+L^\perp ) \cap K)\ 
      dz_1\, dy_1^L\, dz_2\, dy_2^L\,  dx
\nonumber\\[1ex]
&\leq
\d_t^{3d+4} \!\!
\int\limits_{ B_d \times (2 B_2 )^2  \times B_d^2}  \!\!
      \1([0,x\vert_L ] \cap (y_1^L + [0,z_1\vert_L]) \neq \emptyset) 
      \vol_{d-2} ((y_1^L+L^\perp ) \cap K) 
\\ & \hskip3cm 
      {\1}([0,x\vert_L ] \cap (y_2^L + [0,z_2\vert_L]) \neq \emptyset) 
      \vol_{d-2} ((y_2^L+L^\perp ) \cap K)\ 
\\ & \hskip3cm 
      dz_1\, dy_1^L\, dz_2\, dy_2^L\,  dx
\\[1ex] &\leq
c'_d \d_t^{3d+4} \bigl(\max_{u \in \d_t B_d } \vol_{d-2} ((u +L^\perp ) \cap K) \bigr)^2
\end{align*}
with
\begin{align*}
c'_d &=
\int\limits_{ B_d \times (2 B_2 )^2  \times B_d^2} 
\1([0,x\vert_L ] \cap (y_1^L + [0,z_1\vert_L]) \neq \emptyset)  \cdot 
\\
&\hspace{2.8cm}
 \1([0,x\vert_L ] \cap (y_2^L + 
[0,z_2\vert_L]) \neq \emptyset) 
dz_1\,  dz_2\, dy_1^L\, dy_2^L\,  dx.
\end{align*}
Analogously, 
$$
J^{(2)}_L (K)
\geq
c'_d \d_t^{2d+2}  \1( 2 \d_t B_2 \subset K\vert_L)
\bigl(\min_{u \in \d_t B_d } \vol_{d-2} ((u +L^\perp ) \cap K)\bigr)^2 .
$$ 
As in the proof of Lemma \ref{le:J1K} we calculate 
\begin{align*}
c'_d 
&=
\int\limits_{ B_d \times (2 B_2 )^2  \times B_d^2} 
\1([0,x\vert_L ] \cap (y_1^L + [0,z_1\vert_L]) \neq \emptyset)  \cdot 
\\
&\hspace{2.8cm}
 \1([0,x\vert_L ] \cap (y_2^L + 
[0,z_2\vert_L]) \neq \emptyset) 
dz_1\,  dz_2\, dy_1^L\, dy_2^L\,  dx 
\\ &=
\kappa_{d-2}^3  \int\limits_{ B_{2}^3} 
V_2( [0,x^L ] + [0,-z_1^L]) V_2( [0,x^L ] + [0,-z_2^L])
(1- \|z_1^L\|)^{\frac {d-2}2} \cdot 
\\
&\hspace{2.8cm} (1- \|z_2^L\|)^{\frac {d-2}2} (1- \|x^L\|)^{\frac {d-2}2}
dz_1^L\,  dz_2^L\, dx^L .
\end{align*}
Changing to polar coordinates gives
\begin{align*}
c'_d &=
\kappa_{d-2}^3  \int\limits_{ S_1^3}\int\limits_{ [0,1]^3} 
|{\rm det} (r_1 u_1 , r_2 u_2) | \ |{\rm det} (r_1 u_1 , r_3 u_3) |
(1- r_1)^{\frac {d-2}2}  \cdot 
\\
&\hspace{2.8cm} (1- r_2)^{\frac {d-2}2} (1- r_3)^{\frac {d-2}2}
r_1 r_2 r_3 \, dr_1 \, dr_2 \, dr_3 \, du_1 \, du_2 \, du_3 
\\ &=
\kappa_{d-2}^3  \int\limits_{ [0,1]^3} r_1^3 r_2^2 r_3^2
(1- r_1)^{\frac {d-2}2} (1- r_2)^{\frac {d-2}2} (1- r_3)^{\frac {d-2}2} \, dr_1 \, dr_2 \, dr_3 
\cdot \\
&\quad
\int\limits_{ S_1^3} |{\rm det} (u_1 ,u_2) | \ |{\rm det} (u_1 ,u_3) |\, du_1 \, du_2 \, du_3 
\\ &=
\kappa_{d-2}^3  \, {\bf B}\!\left( 3,\frac d2 \right)^2  \, {\bf B}\!\left( 4,\frac d2 \right)
\int\limits_0^{2\pi}
(\cos(\alpha))^2 \, d\alpha  
\\ &=
\pi \kappa_{d-2}^3  \, {\bf B}\!\left( 3,\frac d2 \right)^2  \, {\bf B}\!\left( 4,\frac d2 \right).
\end{align*}
This  proves the lemma.
\end{proof}

\begin{proposition}\label{prop:Int2-rcr}
Let $v$ be a point in $W\subset \R^d$. Then for
\begin{eqnarray*}\label{defI2_W(v)}
I^{(2)}_{W,L}(v)
&=&
 \int\limits_{W^5} 
\1([v,v_2]\vert_L \cap [v_3,v_4]\vert_L \neq \emptyset, \,  \| v-v_2\| \leq \d_t,\ \| v_3-v_4\| \leq \d_t)
\\[-3ex]&&  \hskip7mm 
\1([v,v_2]\vert_L \cap [w_3,w_4]\vert_L \neq \emptyset, \, 
\| w_3-w_4\| \leq \d_t)
\ d v_2 d v_3 d v_4 d w_3 d w_4 
\end{eqnarray*}
it holds that
\begin{equation}\label{eq:lim-I2W}
\lim_{\d_t \to 0} \ \frac{I^{(2)}_{W,L}(v)}{\d_t^{3d+4} } 
= c'_d 
\vol_{d-2} ((v +L^\perp ) \cap W)^2 
\end{equation}
as $\d_t \to 0$, with $c'_d$ given in \eqref{def:cd2}.
Further we have 
$ I^{(2)}_{W,L}(v) \leq c'_d \d_t^{3d+4} M_{d-2}^2 (W)  $ 
where $M_{d-2}(W)$ is the volume of the maximal $(d-2)$-dimensional section of $W$.  
\end{proposition}

\begin{proof}
We substitute $v_2= v + x $, $v_3 = v + y_1$ and $v_4 = v + y_1+ z_1$, and  $w_3 = v + y_2$
 and $w_4 = v+y_2+z_2 $, respectively. We obtain
\begin{eqnarray*}
I^{(2)}_{W,L}(v) & = &
\int\limits_{\substack{(W-v)^3 \\\times(W-v-y_2) \\\times(W-v-y_1)}} 
\1([0,x]\vert_L \cap (y_1+[0,z_1])\vert_L \neq \emptyset,
\, \|x\| \leq \d_t,\ \|z_1\| \leq \d_t)
\\[-7ex]&&  \hskip17mm 
\1([0,x]\vert_L \cap (y_2 +[0,z_2])\vert_L \neq \emptyset, 
\, \| z_2\| \leq \d_t) \, 
 d z_1 d z_2 d y_1 d y_2 dx 
\\[5ex] & \leq &
\int\limits_{(\d_t B_d) \times (W-v)^2 \times (\d_t B_d)^2} 
\1([0,x]\vert_L \cap (y_1+[0,z_1])\vert_L \neq \emptyset) 
\\[-3ex]&& \hskip33mm 
 \1([0,x]\vert_L \cap (y_2 +[0,z_2])\vert_L 
\neq \emptyset) 
 d z_1 d z_2 d y_1 d y_2 dx 
\\[1ex] &=& 
J_L^{(2)} (W-v)
. 
\end{eqnarray*}
Analogously to the calculations in Proposition \ref{prop:Int-rcr}, the expectation is also bounded 
from below by
$$
J_L^{(2)} (W_{- \d_t}-v)
. $$
Together with Lemma \ref{le:J2K} this shows that the expectation equals
\begin{equation}\label{eq:Int2rcr}
 c'_d \d_t^{3d+4} \vol_{d-2}(L^\perp \cap W)^2 (1+\co(1)) . 
\end{equation}

\end{proof}

\subsection{Projecting the RGG on a Fixed Plane}
In the next sections we fix an arbitrary two-dimensional plane $L$ and investigate $\rcr (G_L)$ by projecting on this fixed plane.

\subsubsection{The Expectation of the Rectilinear Crossing Numbers}\label{sec:Ecr}
The expectation of the rectilinear crossing numbers with respect to the underlying Poisson point process can be computed using the  Slivnyak-Mecke formula \eqref{eq:Mecke} and by some integralgeometric investigations.

\begin{theorem}\label{th:exp-rcr}
Let $ G_L$ be the projection of the RGG onto a two-dimensional plane $L$. Then, 
as $t \to \infty$ and $\d_t \to 0$,
$$
 \E_{V} \rcr (G_L)
= \frac 18 c_d \,  t^4\d_t^{2d+2}  \,  I^{(2)}_{W}(L) + \co(\d_t^{2d+2} t^4) ,
$$ 
where $I^{(2)}_{W}(L):= \int_{W\vert_L} \vol_{d-2} ((v +L^\perp ) \cap W)^2  \, dv $ and $c_d$ is defined in \eqref{def:cd1}.
\end{theorem}

\begin{proof}
Recalling the definition \eqref{defI_W(v)}, the Slivnyak-Mecke formula \eqref{eq:Mecke} yields 
\begin{equation*}
	\E_{V} \rcr (G_L) =\frac 18 \ t^4 \int\limits_W I^{(1)}_{W,L}(v_1) \;dv_1.
\end{equation*}
Now we apply \eqref{eq:lim-IW} in Proposition \ref{prop:Int-rcr}.
Using the dominated convergence theorem of  Lebesgue and 
Fubini's theorem we obtain
\begin{align*}
\lim_{t \to \infty} \frac{ \E_V \rcr (G_L)}{t^4 \d_t^{2d+2} }
& =  
\frac 18 c_d \int\limits_{W} \vol_{d-2} ((v_1 +L^\perp ) \cap W)  \, dv_1 
\\ & =  
\frac 18 c_d \int\limits_{W\vert_L} \int\limits_{L^\perp} \vol_{d-2} ((v_1 +L^\perp ) \cap W)  \, dv^{L^\perp}_1 
 dv^L_1 
\\ & =  
\frac 18 c_d \int\limits_{W\vert_L} \vol_{d-2} ((v_1^L +L^\perp ) \cap W)^2  \, dv^{L}_1 .\qedhere
\end{align*}
\end{proof}

For unit-disk graphs, i.e., $d=2$, the choice of $L$ is unique and the projection superfluous.
There the expected crossing number is asymptotically
$  \frac {c_2}8\,  t^4 \d_t^{6} $ and thus of order 
$\Theta( {m^3}/{n^2} )$
which is asymptotically optimal as witnessed by the crossing lemma.
In general, the expectation is of order
$$ 
t^4 \d_t^{2d+2} 
=\Theta\left(
\frac{m^3}{n^2} \left( \frac{m}{n^2} \right)^{\frac {2-d}d} \right).
$$
The extra factor $m/n^2$ can be understood as the probability that two vertices are connected via an 
edge,
thus measures the ``density'' of the graph.

\subsubsection{The Variance of the Rectilinear Crossing Numbers}\label{sec:var-rcr}
The method to calculate the variance of U-statistics is by expanding the second moment in a suitable way and using the Slivnyak-Mecke formula. This leads to various sums over $k$-tuples of points where the asymptotic dominating summands have to be identified.
For simplicity we concentrate on the case $d \geq 3$ since in the two-dimensional case the number of dominating terms is larger (although the computations are similar).
\begin{theorem}\label{th:var-rcr}
Let $ G_L$ be the projection of an RGG in $\R^d$, $d \geq 3$, onto a two-dimensional plane 
$L$. Then, as $t \to \infty$ and $\d_t \to 0$,
\begin{equation}\label{eq:var-rcr}
\lim_{t \to \infty} 
\frac{\V_V \rcr (G_L)}{t^7 \d_t^{4d+4}} 
= 
\frac 18 \left(2 c_d^2 + c'_d \lim_{t\to \infty} \frac 1{t\d_t^{d}}\right) I^{(3)}_{W}(L) 
= 
\frac 18 c''_d I^{(3)}_{W}(L) 
\end{equation}
where $I^{(3)}_{W}(L):=\int_{W\vert_L}  \vol_{d-2} ((v +L^\perp ) \cap W)^3\, dv $, $c_d$ is defined in \eqref{def:cd1} and $c'_d$ in \eqref{def:cd2}.
\end{theorem}

\begin{proof}
For a moment we use the abbreviation 
$$f(v_1, \dots, v_4 )= \1([v_1,v_2]\vert_L \cap [v_3,v_4]\vert_L \neq \emptyset,\,
	  \| v_1-v_2\| \leq \d_t,\ \| v_3-v_4\| \leq \d_t). $$ 
The second moment of the crossing number is given by
\begin{equation} \label{eq:E-cross^2}
\E_V  \rcr (G_L)^2 = 
\frac 1{64}
\E_{V} \Big( \sum_{(v_1, \dots, v_4) \in V_{\neq}^4} f(v_1, \dots, v_4 ) \Big)^2 .
\end{equation}
The square of the sum yields the sum over two 4-tuples $(v_1, \dots v_4)$ and $(w_1, \dots , w_4)$ which may or may not overlap. The sum over those 4-tuples which are disjoint yields by the Slivnyak-Mecke formula $(\E _V \rcr (G_L))^2$. Thus by the definition of the variance
\begin{equation} \label{eq:Var-cross}
\V_V  \rcr (G_L)^2 = 
\frac 1{64} \E_{V} 
\sum_{\substack{(v_1, \dots, v_4) \in V_{\neq}^4 \\ (w_1, \dots, w_4) \in V_{\neq}^4\\ 
| \{ v_1, \dots, v_4\} \cap \{w_1, \dots, w_4\}| \geq 1  }} 
f(v_1, \dots, v_4 ) f(w_1, \dots, w_4 ) .
\end{equation}
We start by computing the case when only one point coincides, that is, $| \{ v_1, \dots, v_4\} \cap \{w_1, \dots, w_4\}| = 1$. There are 16 possibilities to choose two points from each set. 
\begin{eqnarray*}  
\frac 1{64} \E_{V} \lefteqn{
\sum_{\substack{(v_1, \dots, v_4) \in V_{\neq}^4 \\ (w_1, \dots, w_4) \in V_{\neq}^4\\ 
| \{ v_1, \dots, v_4\} \cap \{w_1, \dots, w_4\}| = 1  }} 
f(v_1, \dots, v_4 ) f(w_1, \dots, w_4 ) 
= }&&
\\ &=&
\frac 1{4} \E_{V} 
\sum_{(v,v_2 \dots, v_4,w_2, \dots, w_4) \in V_{\neq}^7} 
f(v, v_2, \dots, v_4 ) f(v,w_2, \dots, w_4 ) 
\\ &=&
\frac 1{4} t^7 
\int_W \int_{W^6} f(v, v_2, \dots, v_4 ) f(v,w_2, \dots, w_4 ) dv_2 \dots dw_4\, dv
\end{eqnarray*}
and the inner integral equals $I^{(1)}_{W,L}(v)^2$. Now we apply again \eqref{eq:lim-IW} in Proposition \ref{prop:Int-rcr}. Using the dominated convergence theorem of  Lebesgue and 
Fubini's theorem we obtain
\begin{eqnarray*}  
\frac 1{64} \E_{V} \lefteqn{
\sum_{\substack{(v_1, \dots, v_4) \in V_{\neq}^4 \\ (w_1, \dots, w_4) \in V_{\neq}^4\\ 
| \{ v_1, \dots, v_4\} \cap \{w_1, \dots, w_4\}| = 1  }} 
f(v_1, \dots, v_4 ) f(w_1, \dots, w_4 ) 
= }&&
\\ &=&
\frac 1{4} c_d^2  t^7 \d_t^{4d+4}
\int_{W\vert_L} \vol_{d-2} ((v^L +L^\perp ) \cap W)^3 \, dv^L \ (1 +\co(1))
\\ &=&
\frac 1{4} c_d^2  t^7 \d_t^{4d+4} 
I^{(3)}_W(L) \ (1 +\co(1)) .
\end{eqnarray*}

For all other terms in the expectation $\E_{V }$, at least two points of the sets $\{v_2, \dots, 
v_4\}$, $\{w_2, \dots, w_4\}$ coincide. 
We first calculate the situation where two segments coincide. Since there are two segments involved in each function and each segment consists of two points there are 8 possibilities. We calculate in particular the case $v_1=w_1$ and $v_2 = w_2$. 
The Slivnyak-Mecke formula gives
\begin{eqnarray*}
\frac 1{8} \lefteqn{\E_{V } 
\sum_{v_1, \dots, v_4,w_3, w_4 \in V_{\neq}^6} f(v_1, \dots, v_4) f(v_1, v_2, w_3, w_4) } &&
\\ &=&
\frac 18 t^6 \int\limits_W \int\limits_{W^5} 
f(v,v_2\dots ,v_4) f(v, v_2, w_3, w_4)
\ d v_2 d v_3 d v_4 d w_3 d w_4 \,dv ,
\end{eqnarray*}
and the inner integral equals $I^{(2)}_{W,L}(v)^2$. This time we apply \eqref{defI2_W(v)} in Proposition \ref{prop:Int2-rcr}. Using the dominated convergence theorem of Lebesgue and 
Fubini's theorem we obtain
\begin{eqnarray*}
\frac 1{8} \lefteqn{\E_{V } 
 \sum_{v_1, \dots, v_4,w_3, w_4 \in V_{\neq}^6} f(v_1, \dots, v_4) f(v_1, v_2, w_3, w_4) } &&
\\ &=&
\frac 18 c'_d t^6 \d_t^{3d+4} \int_{W\vert_L} \vol_{d-2} ((v^L +L^\perp ) \cap W)^3 \, dv^L  \ (1 +\co(1))
\\ &=&
\frac 18 c'_d t^6 \d_t^{3d+4} I^{(3)}_{W}(L)  \ (1 +\co(1)).
\end{eqnarray*}

Similar investigations show that in the case where two points in different segments coincide we obtain terms of order 
$\cO(t^{6} \d_t^{4 d +  2})$.
The case when three points coincide yields terms of order at most $\cO(t^{5} \d_t^{3 d + 2})$. If all points of the two 4-tuples coincide we obtain the expectation of the crossing numbers which is $\cO(t^{4} \d_t^{2 d + 2})$.

Combining our results  we have
\begin{eqnarray*}
\V_{V } \rcr (G_L) 
& = & \nonumber
\frac 18 t^7 \d_t^{4d+4} (2 c_d^2 + c'_d t^{-1} \d_t^{-d}) I^{(3)}_{W}(L)  \ (1 +\co(1)) 
\\ && +
\cO(t^{6} \d_t^{4 d +  2}) + \cO(t^{5} \d_t^{3 d + 2})+  \cO(t^{4} \d_t^{2 d + 2})
\ . 
\end{eqnarray*}
Since we assume $d \geq 3$ and that $ t \d_t^d \geq c$, this leads to \eqref{eq:var-rcr}.
\end{proof}

\noindent
Theorem \ref{th:exp-rcr} and Theorem \ref{th:var-rcr} show for the standard deviation
$$
\s (\rcr (G_L)) = \sqrt{\V_V \rcr (G_L)} = \Theta(
t^4 \d_t^{2d+2} \  t^{-\frac 12} ) = \Theta(
\E_V \rcr (G_L)\ t^{-\frac 12} ) ,
$$
which is smaller than the expectation. Or, 
equivalently, the coefficient of variation 
$\frac {\s(\rcr(G_L))}{\E_V \rcr(G_L)} $
is of order $t^{- \frac 12}$.
As $t \to \infty$, our bounds on the expectation and variance together with Chebychev's inequality lead to 
\begin{eqnarray*}
\P\left(\left\vert \frac{\rcr(G_L)}{t^4 \d_t^{2d+2}}- \frac{\E_V \rcr(G_L)}{t^4 
\d_t^{2d+2}}\right\vert \geq \e \right) \leq \frac{\V_V \rcr(G_L)}{t^8 \d_t^{4d+4} \e^2}  \to 0.
\end{eqnarray*}
Recalling $I^{(2)}_{W}(L)$ defined in Theorem \ref{th:exp-rcr}, we can state
\begin{corollary}[Law of Large Numbers]\label{cor:llnrcr}
For given $L$, the normalized random crossing number converges in probability (with respect to the 
Poisson point process $V$) as $t \to \infty$,
$$\frac{\rcr(G_L)}{t^4 \d_t^{2d+2}} \ \to \ \frac 18 c_d I^{(2)}_{W}(L) .
$$
\end{corollary}

\subsubsection{Deviation Inequalities}

As known by the crossing lemma, the optimal crossing number is of order $\Omega(\frac {m^3} {n^2})$. In 
our setting this means that we are looking for crossing numbers of order $t^4 \d_t^{3d}$, much smaller than the expectation $\E_V \rcr (G_L) $.
Chebychev's inequality shows that it is difficult to reach this order of magnitude; using $\d_t \to 0$ in the last step we have:
\begin{align*}
\P_{V}(\rcr (G_L) \leq c t^4\d_t^{3d})
&\leq
\P_{V}\bigl(|\rcr (G_L) - \E_{ V} \rcr (G_L) | \!\geq \E_{ V} \rcr (G_L) - 
c t^4\d_t^{3d}\bigr)
\\
&\leq
\frac {\V_{V } \rcr (G_L) }{(\E_{ V } \rcr (G_L) - c t^4\d_t^{3d})^2}
=
\cO(t^{-1}).
\end{align*}

To get a more precise inequality we use a large deviation inequality from \cite{LachRei}, which follows from a general isoperimetric inequality on Poisson spaces for the convex distance, see \cite{Reitzner13, CTMRMS}. 
Let $U=U_k(f)$ be a U-statistic of order $k$, and denote by $U(v)$ the local version of the U-statistic $U$, 
$$
U(v)= \sum_{(v_2, \dots, v_{k}) \in (V\setminus\{v\})^{k-1}_{\neq}} f(v,v_2, \dots, v_{k})   
$$
so that $ U= \sum_{ v\in V} U(v) $. 

\begin{theorem}[see {\cite[Theorem 8]{LachRei}}]\label{th:LDIgen}
Let $U=U_k(f)$ be a Poisson U-statistic of order $k$, and let $\M U$ be the median of $U$.
Then for $u\geq 0$ and arbitrary $B>0$, 
 \begin{equation}
\P( U \leq \M U - u)
\leq
2 \exp\left(- \frac {u^2 }{ 4 k^2 B\,   \M U}\right)
+  2 \P(\exists v:\  U(v) >B) \,.
\end{equation}
\end{theorem}

\noindent We need to compute a uniform upper bound for $\P (U(v) > B)$ since
\begin{eqnarray*}
\P(\exists v:\  U(v ) >B) 
&\leq &
\E \sum_{v\in V } \1 (U(v) > B)
\\ & = &
t \int_{W} \P (U(v) > B) dv 
\\ & \leq &
t \max_v \P (U(v) > B).
\end{eqnarray*}
To apply this for the crossing numbers observe that in our case we investigate $U=U_4(f)= \rcr(G_L)$ with  
\begin{eqnarray*}
f(v, v_2,v_3, v_4)
&=&
\1([v,v_2]\vert_L \cap [v_3, v_4]\vert_L \neq \emptyset, \!\| v-v_2\| \leq \d_t,\ \| v_3-v_4\| \leq \d_t) 
\\ &\leq &
\1(v_2 \in B(v, \d_t),\, v_3 \in B_L(v , 2\d_t) + L^\perp,\, v_4 \in B(v_3, \d_t)) .
\end{eqnarray*}
To gain more independence we distinguish whether $B(v, \d_t)$ is disjoint from $B(v_3, \d_t)$ or not.
\begin{align*}
U(v ) = &
\ \rcr(G_L) (v)
\\ \leq\ &
\sum_{V^3_{\neq}} \1(v_2, v_3, v_4  \in B(v, 2\d_t))\\ 
&{}+ \sum_{V^3_{\neq}}  \1(v_2 \in B(v, \d_t),\, v_3 \in ((B_L(v , 2\d_t) + L^\perp)\setminus B(v, 2\d_t),\, v_4 \in B(v_3, \d_t))  .
\end{align*}
Denote by $N_1$ a Poisson distributed random variable with mean $E_1= t (2\d_t)^d \kappa_d$,
by $N_2, N_4$ independent Poisson distributed random variables with mean $E_2=E_4= t \d_t^d \kappa_d$, and
by $N_3$ an independent Poisson distributed random variable with mean $E_3=  t (2\d_t)^2 \pi M_{d-2}(W)  $, 
where $M_{d-2}(W)$ is the maximal $(d-2)$-dimensional section of $W$.
Because $\rcr(G_L)(v)$ is stochastically dominated by 
$ N_1^3 + N_2 N_3 N_4$, 
we have
\begin{eqnarray*}
\max_v \P( \rcr(G_L)(v ) >B)
& \leq &
\P( N_1^3 + N_2 N_3 N_4 >B )
\\ & \leq &
\P \left( N_1^3 > \frac 12 B \right)
 + \P \left( N_2 N_3 N_4 > \frac 12 B \right)
\\ & \leq &
\P \left( N_1 > \Bigl(\frac 12 B\Bigr)^{\frac 13} \right)
 + \prod_{i=2,3,4} \P \left( N_i > \Bigl(\frac 12 B\Bigr)^{\frac 13} \right).
\end{eqnarray*}
It remains to estimate the two summands. Starting with the estimate for the probability that 
$N_i > (\frac 12 B )^{\frac 13}$, 
we use the Chernoff bound for the Poisson distribution, namely
\begin{equation} \label{eq:chern}
\P(N_i > r) \leq \inf_{s\geq 0} e^{E_i(e^s-1)-sr} .
\end{equation}
Since 
$$
\inf_{s\geq 0} {E_i(e^s-1)-sr}
= r(1-\ln \frac r{E_i} ) -{E_i}  \leq - r 
$$
for $r \geq E_i e^2$, we obtain the estimate
\begin{eqnarray*}
\P \left(N_1^3 > \frac 12 B \right)
+ 
\prod_{i=2,3,4} \P \left(N_i > \Bigl(\frac 12 B\Bigr)^{\frac 13} \right)
&\leq &
e^{  - (\frac 12 B)^{\frac 1 3} }  
+
e^{  - 3(\frac 12 B)^{\frac 1 3} }  
\end{eqnarray*}
for $B \geq \max (2E_i^3 e^6)$.
We combine this with the general Theorem \ref{th:LDIgen},
 \begin{equation}
\P( \rcr(G_L) \leq \M \rcr(G_L) - u)
\leq
2 \exp\left(- \frac {u^2 }{ 4^3 B \, \M \rcr(G_L)}\right)
+  4 t e^{  - (\frac 12 B)^{\frac 1 3} }  
\end{equation}
for $B \geq 2 E_3^3 e^6$, because $ \max E_i= E_3$ for $\d_t$ sufficiently small. We set $ B = \frac{1}{4^2} (\frac {u^{2} }{  \M \rcr(G_L)})^{\frac{3}4} $.
\begin{theorem}\label{th:localLDI}
For $ u \geq 32^{2/3} \left( t (2 \d_t)^2 \pi M_{d-2}(W)\right)^{2} e^4 \sqrt{\M \rcr(G_L)}$ we have
\begin{eqnarray*}
\P( \rcr(G_L) \leq \M \rcr(G_L) - u) 
\leq (2+4 t)  e^{- \frac{1}{4} \left(\frac {u^{2} }{ \M \rcr(G_L)}\right)^{\frac{1}4} }  .
\end{eqnarray*}
\end{theorem}

Since the difference between median and expectation satisfies 
$$
\left| \frac {\M \rcr(G_L)}{\E \rcr(G_L)} -1\right| \leq \frac{\sqrt{\V \rcr(G_L)}}{ \E \rcr(G_L)}  = \cO(t^{- \frac 12}),
$$
we see that 
\begin{equation}\label{eq:Masymp}
 \M \rcr(G_L) = t^4 \d_t^{2d+2} I^{(2)}_{W}(L)) + \co(t^4 \d_t^{2d+2}) .  
\end{equation}
For the next theorem we set $u= \M \rcr(G_L) - c_1 t^4 \d_t^{3d} $, and take into account that in our case $ t \d_t^d \geq c$, $d >2$, and thus 
$ t^4 \d_t^{2d+2}  \geq  c t $.

\begin{theorem}\label{thm:manytries}
Let $ G_L$ be the projection of the RGG onto a two-dimensional 
plane $L$. Then there are $c_1, c_2 >0$ such that the random crossing number satisfies
$$\P\big( \rcr(G_L) \leq c_1 t^{4}\d_t^{3d}\big) 
\leq
e^{- c_2 \sqrt[4]t  }  $$
for $t \geq 1$.
\end{theorem}

\bigskip
\subsection{Projecting the RGG on a Random Plane}

Until now, we fixed a plane $L$ and computed the variance with respect to the random points $V$.
In general, Theorem \ref{th:exp-rcr} and Theorem \ref{th:var-rcr} allow to compute the expectation 
and variance 
with respect to $V$ \emph{and} the randomly chosen plane $L$. For 
the expectation we obtain from Theorem \ref{th:exp-rcr} and by Fubini's theorem
\begin{equation}\label{eq:exp-L-ppp-allg}
\E_{L, V} \rcr(G_L)
=
\frac {1}8 c_d \, t^4\d_t^{2d+2}\int\limits_{\cL} I^{(2)}_{W}(L)  \, dL\  
 + \co( t^4 \d_t^{2d+2}),
\end{equation}
as $t \to \infty$ and $\d_t \to 0$, where $dL$ denotes integration with respect to the Haar measure on~$\cL$.
We use the variance decomposition 
$  \V_{L, V} X=  \E_L \V_V X + \V_L \E_V X$. 
By
\begin{align} 
 \E_L \V_V \rcr(G_L) &=
\frac {1}{8} c''_d \, t^7 \d_t^{4d+4} 
\int\limits_{\cL} I^{(3)}_{W}(L) \, dL \ 
+\co(t^7 \d_t^{4d+4} ),\text{\quad and}\\
 \V_L \E_V \rcr(G_L) 
& = 
\E_L (\E_V \rcr(G_L))^2 -  (\E_{L,V} \rcr(G_L))^2 = \\
\frac {1}{64} c_d^2  \, t^8\d_t^{4d+4} & \left[ 
\int\limits_{\cL} I^{(2)}_{W}(L)^2 \, dL - \left(\int\limits_{\cL} I^{(2)}_{W}(L) 
dL \right)^2 
\right]
+ \co(t^8 \d_t^{4d+4} )\nonumber
\end{align}
we obtain 
\begin{align}\label{eq:totalVar}
\V_{L,V} \rcr(G_L)
=\ & 
\frac {1}{64} c_d^2\, t^8\d_t^{4d+4} \left[ 
\int\limits_{\cL} I^{(2)}_{W}(L)^2 dL - \left(\int\limits_{\cL} I^{(2)}_{W}(L) 
dL \right)^2 
\right]
\\
& + \co(t^8 \d_t^{4d+4} ) . \nonumber
\end{align}
It is not a surprise that the variance $\V_{L, V} \rcr (G_L)$ is dominated by $ \V_L \E_V \rcr(G_L) $ in general. By Theorem \ref{thm:manytries}, the crossing numbers are sharply concentrated around their mean which depends on $L$ via $I^{(2)}_{W}(L)$. Hence the variance is just dominated by the variance of the random variable $I^{(2)}_{W}(L)$ for random $L \in \cL$, 
$$ \V_{L} I^{(2)}_{W}(L) =  
\int\limits_{\cL} I^{(2)}_{W}(L)^2 \, dL - \left(\int\limits_{\cL} I^{(2)}_{W}(L) 
dL \right)^2 $$
H\"older's inequality implies that this is positive as long as $I^{(2)}_{W}(L)$ is not a 
constant function.

{\it Remark.} The integral appearing in the expectation can be rewritten using integration over $\cE^d_{d-2}$, the set of affine $(d-2)$-dimensional planes with respect to the corresponding Haar measure, which gives
$$
\int\limits_{\cL} I^{(2)}_{W}(L)  \, dL
=
\int\limits_{\cE^d_{d-2}} \vol_{d-2}(E \cap W)^2  \, dE  .
$$
It would be of interest to investigate the minimizers, resp.\ extremizers, of this integral with respect to $W$.

\subsubsection{The Rotation Invariant Case}\label{sec:isotropic-rcr}
If $W$ is the ball $B$ of unit volume and thus $V$ is \emph{rotation invariant}, 
then 
$I^{(2)}_{B}(L)=I^{(2)}(B)$ is a constant function independent of $L$, and the leading term in \eqref{eq:totalVar} is vanishing.
From  \eqref{eq:exp-L-ppp-allg} we see that in this case  the expectation is independent of~$L$.
$$
\E_V \rcr(G_L) = \E_L \E_V \rcr (G_L)
= t^4  \d_t^{2d+2}  I^{(2)}(B)+ \co(t^4 \d_t^{2d+2} ) 
$$
For the variance this implies
$\V_L \E_V \rcr (G_L)=0$, and hence 
\begin{eqnarray*} 
\V_{L, V} \rcr (G_L) 
&=&
\E_L \V_V \rcr (G_L) 
=
\frac {1}{8} c''_d
\, t^7 \d_t^{4d+4} 
I^{(3)}(B) 
+\co(t^7 \d_t^{4d+4} ).
\end{eqnarray*}
In this case the variance $\V_{L,V}$ is of the order $t^{-1}$---and thus surprisingly 
significantly---smaller than in the general case.

\begin{theorem}\label{th:isotropic-rcr}
Let $ G_L$ be the projection of an RGG in the ball $B \subset \R^d$, $d \geq 3$, onto a 
two-dimensional uniformly chosen random plane 
$L$. Then
\begin{align*}
 \E_{L,V} \rcr (G_L)
&= \frac 18 c_d  \, t^4  \d_t^{2d+2}  I^{(2)}(B) 
(1 + \co(1) )\text{\quad and }\\
\V_{L,V} \rcr (G_L) 
&= \frac 1{8} c''_d\,  t^7 \d_t^{4d+4}    I^{(3)}( B)
(1 + \co(1) ),
\end{align*}
as $t \to \infty$.
\end{theorem}

Again, Chebychev's inequality immediately yields a law of large numbers which states that with 
high probability the crossing number of $G_L$ in a random direction is very 
close to 
$ \frac 18 c_d \, t^4 \d_t^{2d+2} I^{(2)}(B) $.
\begin{corollary}[Law of Large Numbers] \label{cor:llninv}
Let $ G_L$ be the projection of an RGG in $B \subset \R^d$, $d \geq 3$, onto a random two-dimensional 
plane $L$. Then the normalized random crossing number converges in probability (with respect to the 
Poisson point process $V$ and to $L$), as $t \to \infty$,
$$\frac{\rcr(G_L)}{t^4 \d_t^{2d+2}} \ \to \ \frac 18 c_d I^{(2)}(B) .
$$
\end{corollary}

\subsubsection{Deviation Inequalities}
To obtain a deviation inequality for fixed $W$ and randomly chosen $L$ and thus $G_L$, recall that by \eqref{eq:Masymp}, and since $I^{(2)}_{W}(L)$ is bounded from below and above, there are constants $\underline c, \overline c >0$ depending on $W$ such that the median satisfies
$$ \underline c t^4 \d_t^{2d+2} \leq \M \rcr(G_L) \leq \overline c t^4 \d_t^{2d+2} $$
for $t \geq 1$. Theorem \ref{th:localLDI} then tells us that for $ u \geq 4 \pi  32^{2/3} e^4 M_{d-2}(W)^{2}  \sqrt{\overline c}\ t^4 \d_t^{d+3}  $ we have
\begin{eqnarray*}
\P\big( \rcr(G_L) \leq \underline c t^4 \d_t^{2d+2} - u\big) 
\leq (2+4 t)  e^{- \frac{1}{4} \left(\frac {u^{2} }{ \overline c t^4 \d_t^{2d+2}}\right)^{\frac{1}4} }  .
\end{eqnarray*}

For the next theorem choose some $c_1 >0$. We set $u= \underline c t^4 \d_t^{2d+2} - c_1 t^4 \d_t^{3d} $ which is positive for $t \geq T$, and take into account that in our case $ t \d_t^d \geq c$, $d >2$, and thus 
$ t^4 \d_t^{2d+2}  \geq  c t $.

\begin{theorem}\label{thm:manytriesinL}
Choose some $c_1 >0$. Let $ G_L$ be the projection of the RGG onto a random two-dimensional 
plane $L$. Then the random crossing number satisfies
$$\P\big( \rcr(G_L) \leq c_1 t^{4}\d_t^{3d}\big) 
\leq
e^{- c_2 \sqrt[4]t  }  $$
for $t \geq T$ where $c_2$ and $T$ depend on $c_1$ and $W$.
\end{theorem}

Hence it seems to be expensive to use the following computational na\"ive approach: try to minimize the crossing numbers by just 
projecting onto a sequence of random planes until reaching the theoretical minimum of order $\cO(t^4 \d_t^{3d})$.
This suggests to combine the search for an optimal choice of the direction of projection with other 
quantities of the RGG. It is a long standing assumption in graph drawing that there is a 
connection between 
the crossing number and the stress of a graph. Therefore the next section is devoted to 
investigations concerning the stress of RGGs.

\section{The Stress of an RGG}\label{sec:stress}
According to  \eqref{def:stress} we define the stress of $G_L$ as 
$$
\st(G,G_L)
:= \frac 12 \sum_{(v_1,v_2)\in V_{\neq}^2} w(v_1,v_2) ( d_0(v_1,v_2) - d_L(v_1, v_2))^2,
$$ 
where $w(v_1,v_2)$ is a positive weight-function and $d_0$ resp. $d_L$ are the distances between $v_1$ 
and $v_2$, resp $v_1|_L$ and $v_2|_L$.  
As pointed out in the introduction, typically the weight function is unbounded for $v_2 \to v_1$ but the stress itself is bounded. We need in the following that $\st$ is at least in $L^2$, which for example is guaranteed if there exists an $s >0$ such that 
\begin{equation}\label{eq:stress-bounded}
w(v_1,v_2) ( d_0(v_1,v_2) - d_L(v_1, v_2))^2 \leq s                                                                                                                                                                                                                                          \end{equation}
for all $v_1, v_2 \in W$. We will assume this throughout the paper.
In the following theorem we use the notations
%
%
\begin{align*}
 S^{(1)}_{W,L}(v)&:=\int\limits_{W-v} w(0,v_2) ( d_0(0,v_2) - d_L(0, v_2))^2 dv_2 ,\\
 S^{(1)}_{W}(L)&:=\int\limits_{W} S^{(1)}_{W,L}(v) dv \,  , \mbox{\qquad and\qquad }
 S^{(2)}_{W}(L):=\int\limits_{W} S^{(1)}_{W,L}(v)^2\,  dv.
\end{align*}

\begin{theorem}\label{th:stress}
Let $ G_L$ be the projection of the RGG in $\R^d$, $d \geq 3$, onto a 
two-dimensional plane $L$. Assume that $\st (G,G_L) \in L^2$. Then
\begin{align*} 
\E_{V} \st(G,G_L) 
&= \frac 12 \, t^2 \, S^{(1)}_{W}(L) \text{\quad and} \\
 \V_{V} \st(G,G_L) 
&=  t^3\, S^{(2)}_{W}(L) + \cO(t^2) ,
\end{align*}
\end{theorem}

\begin{proof}
As $\rcr(G)$, stress is a U-statistic, but now of order two. Using the Slivnyak-Mecke formula, it is immediate that 
\begin{eqnarray*}
\E_V \st(G,G_L) 
&=& 
\frac 12 t^2 \int\limits_{W^2} w(v_1,v_2) ( d_0(v_1,v_2) - d_L(v_1, v_2))^2 dv_1 dv_2
\\ &=& 
\frac 12 t^2 \int\limits_{W} \int\limits_{W-v} w(0,v_2) ( d_0(0,v_2) - d_L(0,v_2))^2 dv_2 dv.
\end{eqnarray*}
For the second moment we use $g(v_1, v_2)=w(v_1,v_2) ( d_0(v_1,v_2) - d_L(v_1, v_2))^2$ as an abbreviation. Then the second moment is 
given by
\begin{eqnarray*}
&&\E_V \st(G,G_L)^2 
= 
\frac 14 \E_V \sum_{\substack{(v_1,v_2)\in V_{\neq}^2\\(w_1,w_2)\in V_{\neq}^2}} g(v_1,v_2) g(w_1,w_2) 
\end{eqnarray*}
Similar to the discussion in the case of the crossing numbers, if both pairs are disjoint the expectation yields by the Slivnyak-Mecke formula $(\E_V \st(G,G_L))^2$. This implies that the variance is given by
\begin{eqnarray*}
\V_V \st(G,G_L) 
& = &  
\frac 14 \E_V \sum_{\substack{(v_1,v_2)\in V_{\neq}^2,\, (w_1,w_2)\in V_{\neq}^2\\ |\{v_1, v_2\} \cap \{ w_1, w_2\}| \geq 1}} g(v_1,v_2) g(w_1,w_2) 
\\&=&
\E_V \sum_{(v,v_2, w_2)\in V_{\neq}^3} g(v,v_2) g(v,w_2) 
\\&& + 
\frac 12 \E_V \sum_{(v_1,v_2)\in V_{\neq}^2} g(v_1,v_2)^2
\\&=&
t^3 \int\limits_W g(v, v_2) g(v, w_2) \, dv_2dw_2\, dv 
+ \cO(t^2)
\\&=&
t^3 S^{(2)}_{W}(L) + \cO(t^2)
\end{eqnarray*}
\end{proof}

The discussions from Section \ref{sec:var-rcr} and Section \ref{sec:isotropic-rcr} lead to 
analogous results for the stress of the RGG. Because the standard deviation of the stress is smaller than the 
expectation  by a factor $t^{- \frac 12}$, the stress is concentrated around its mean. 
Using Chebychev's inequality we could derive a law of large numbers. Taking expectations with 
respect to a uniform plane $L$ we obtain:
\begin{align*}
\E_{L, V} \st(G,G_L) 
&= \frac 12 t^2 \int\limits_{\cL} S^{(1)}_{W}(L) dL, \\
\V_{L,V} \st(G,G_L) 
&=
\frac 14  t^4 \left[\int\limits_\cL S^{(1)}_{W}(L)^2 dL - \left(\int\limits_\cL S^{(1)}_{W}(L) dL \right)^2  \right]
+\cO(t^3). 
\end{align*}
Again, the term in brackets is only vanishing if $W=B$.
In this case
$$
\V_{L,V} \st(G,G_L) 
= 
\E_L \V_V \st(G,G_L) 
=
\frac 14  t^3  S^{(2)}_{B}(L)  +\cO(t^2) .
$$

\section{Correlation between Crossing Number and Stress}\label{sec:correlation}

It seems to be widely conjectured that the crossing number and the stress should be positively 
correlated. 
Yet it also seems that a rigorous proof is still 
missing. It is the aim of this section to provide the first proof of this conjecture, in the case 
where the graph is a random geometric graph. 

\subsection{Projecting on a Fixed Plane}
Clearly, by the definition of $\rcr$ and $\st$ we have 
$$
D_v \, \rcr(G_L) \geq 0 \text{ and } D_v \, \st (G, G_L) \geq 0,  $$
for all $v$ and all realizations of $V$. Such a functional $F$ satisfying  $D_v (F) \geq 0$ 
is called increasing. 
The Harris-FKG inequality for Poisson point processes \cite{LastPenrose11} links this fact to the correlation of 
$\rcr(G_L)$ and $\st(G,G_L)$.
\begin{theorem}
Assume that $\st (G,G_L) \in L^2$.  Because $\st$ and $\rcr$ are increasing, we have 
$$\E_V \rcr(G_L) \st (G,G_L) \geq 
\E_V \rcr(G_L) \, \E_V \st (G,G_L) ,
$$
and thus the correlation is positive.
\end{theorem}
We immediately obtain that the covariance is positive and is of order at most
\begin{eqnarray}\label{eq:cov-upperbound}
\cov_V \bigl(\rcr(G_L), \st (G_L) \bigr)
&\leq& \sqrt{\V_V \rcr(G_L)\, \V_V \st (G,G_L) } 
\nonumber 
\\ & & =  \nonumber
\cO( t^{5} \d_t^{2d+2})
\end{eqnarray}
The following theorem proves that this is indeed the correct order.
\begin{theorem}\label{thm:cov}
Let $ G_L$ be the projection of an RGG in $\R^d$, $d \geq 3$, onto a 
two-dimensional plane $L$. Assume as in \eqref{eq:stress-bounded} that the stress is bounded. Then 
$$
\cov_V  \bigl(\rcr (G_L), \st(G,G_L)\bigr) =
\frac 1{2}  c_d  t^5 \d_t^{2d+2} \int\limits_{W} \vol_{d-2} ((v +L^\perp ) \cap W)   S^{(1)}_{W,L}(v) dv \, (1+\co(1))
$$
as $t \to \infty$.
\end{theorem}

\begin{proof}
We again use the abbreviations 
$g(v_1, v_2)=w(v_1,v_2) ( d_0(v_1,v_2) - d_L(v_1, v_2))^2$
and 
$f(v_1, \dots, v_4 )= \1([v_1,v_2]\vert_L \cap [v_3,v_4]\vert_L \neq \emptyset,\, \| v_1-v_2\| \leq \d_t,\ \| v_3-v_4\| \leq \d_t)$.
The covariance is given by
\begin{eqnarray*}
\cov_V  \bigl(\rcr (G_L), \st(G,G_L)\bigr) 
&=& 
\frac 1{16} \E_{V} 
\sum_{\substack{(v_1, \dots, v_4) \in V_{\neq}^4 ,\ (w_1, w_2) \in V_{\neq}^2\\ 
| \{ v_1, \dots, v_4\} \cap \{w_1, w_2\}| \geq 1  }} 
f(v_1, \dots, v_4 ) g(w_1, w_2 ) 
\end{eqnarray*}
since the terms where $(v_1, \dots, v_4)$ is disjoint from $(w_1, w_2)$ by the Slivnyak-Mecke formula yield 
$\E_V \rcr(G_L) \ \E_V \st(G,G_L)$. In the case when both terms have one point in common we have 
\begin{eqnarray*}
\frac 1{16} \E_{V} 
\lefteqn{ \sum_{\substack{(v_1, \dots, v_4) \in V_{\neq}^4 ,\ (w_1, w_2) \in V_{\neq}^2\\ 
| \{ v_1, \dots, v_4\} \cap \{w_1, w_2\}| = 1  }} 
 f(v_1, \dots, v_4 ) g(w_1, w_2 ) }&&
\\&=& 
\frac 1{2} \E_{V} 
\sum_{(v,v_2, \dots, v_4,w_2) \in V_{\neq}^5 } 
f(v,v_1, \dots, v_4 ) g(v, w_2 ) 
\\&=& 
\frac 1{2} t^5 \int\limits_{W} I^{(1)}_{W,L}(v)  S^{(1)}_{W,L}(v) \, dv 
\\&=& 
\frac 1{2}  c_d  t^5 \d_t^{2d+2} \int\limits_{W} \vol_{d-2} ((v +L^\perp ) \cap W)   S^{(1)}_{W,L}(v) \, dv (1+\co(1))
\end{eqnarray*}
by Proposition \ref{prop:Int-rcr}.
For the terms where $w_1, w_2 \in \{ v_1, \dots , v_4\}$ we use that by assumption \eqref{eq:stress-bounded} the stress is bounded by some $s>0$. The Slivnyak-Mecke formula gives
$$
\frac 1{16} \E_{V} 
\sum_{\substack{(v_1, \dots, v_4) \in V_{\neq}^4 \\ \{w_1, w_2\} \subset \{v_1, \dots, v_4\}  }} 
f(v_1, \dots, v_4 ) g(w_1, w_2 )
 \leq  
\frac s4 \E \rcr (G_L) 
=
\cO(t^4 \d_t^{2d+2}) .
$$

\end{proof}

As an immediate consequence we see that the correlation coefficient between the crossing numbers and the stress is bounded away from zero and satisfies
\begin{eqnarray*}
\lim_{t \to \infty} \lefteqn{\cor_V (\rcr(G_L),\st(G,G_L))
} && \\
&=&
\frac{ 
c_d  \int_{W} \vol_{d-2} ((v +L^\perp ) \cap W)   S^{(1)}_{W,L}(v) dv
}{ \sqrt{
\frac 12 c''_d I^{(3)}_{W}(L)   S^{(2)}_{W}(L)
}}
\ .
\end{eqnarray*}
Observe that by H\"older's inequality we always have 
$$ \int_{W} \vol_{d-2} ((v +L^\perp ) \cap W)   S^{(1)}_{W,L}(v) dv
< \sqrt{ I^{(3)}_{W}(L)   S^{(2)}_{W}(L) }
$$
and thus the correlation coefficient is bounded away from one.

\subsection{Projecting on a Random Plane}
The bounds for the covariance in the Poisson point process $V$ given above can be 
used to compute covariance bounds in $L$ and $V$ when $L$ is not fixed but random. For the sake of simplicity we concentrate on the rotation invariant case when $W=B$. As an immediate consequence of Theorem \ref{thm:cov} we obtain
\begin{corollary}\label{cor:corr-VL-isotropic}
Let $ G_L$ be the projection of an RGG in $B \subset \R^d$, $d \geq 3$, onto a 
two-dimensional random plane $L$. Note that by rotation invariance we have $I^{(3)}_{B}(L)=I^{(3)}(B)$ and $S^{(2)}_{B}(L)=S^{(2)}(B)$. Then the correlation between the crossing number and the stress 
of the RGG is positive with
$$
\lim_{t \to \infty} \cor_V (\rcr(G_L),\st(G,G_L))
=
\frac{ 
c_d  \int_{B} \vol_{d-2} ((v +L^\perp ) \cap B)   S^{(1)}_{B,L}(v) dv
}{ \sqrt{
\frac 12 c''_d I^{(3)}(B)   S^{(2)}(B)
}}
$$
\end{corollary}
In particular, the correlation does not vanish as $t \to \infty$. This gives the first proof 
we are aware of, that there is a strict positive correlation between the crossing number and the 
stress of a graph. Hence, at least for rotation invariant RGGs, the method to optimize the stress to obtain good crossing 
numbers can be supported by rigorous mathematics.

\section{Binomial Input}\label{sec:dePoissonize}

Up until now, we preferred to work within the setting of a Poisson point process 
because of the strong mathematical tools that are available in this case. 
However, it is straightforward that the above questions can similarly be investigated in the binomial setting.  As an example we work this out for the computation of the expected crossing numbers in Theorem~\ref{th:exp-rcr}.

Let $V_n $ be a set of $n$ random points uniformly chosen in $W$. Define the graph $G$ by putting an edge between two points  if their distance is bounded by $\d_n$. Here, according to the definitions in Section \ref{sec:tools}, we assume that we are either in the thermodynamic or in the dense regime, which means 
$\lim_{n\to\infty} n \, \d_n^n=c >0 $ or $ \lim_{n\to\infty} n \, \d_n^d=\infty $.
	
We compute the expectation of the rectilinear crossing numbers when projecting on a plane $L$, and use the notation $n_{(k)} = n \dots (n-k+1)$.
By definition,
\begin{eqnarray*}
\lefteqn{\E_{V_n} \rcr (G_L)} \\
&=&
\frac 18 \, \E \sum_{(v_1, \dots, v_4) \in V_{n,\neq}^4} \!\!
\1([v_1,v_2]\vert_L \cap [v_3,v_4]\vert_L \neq \emptyset)  
\1(\| v_1-v_2\| \leq \d_n,\ \| v_3-v_4\| \leq \d_n)
\\ &=&
\frac {n_{(4)}}8  \int\limits_{W^4} \1([v_1,v_2]\vert_L \cap [v_3,v_4]\vert_L \neq \emptyset)  
\1(\| v_1-v_2\| \leq \d_t,\ \| v_3-v_4\| \leq \d_t) \, dv_1 \cdots dv_4
\\ &=&
\frac {n_{(4)}}8 \int\limits_{W} I^{(1)}_{W,L}(v_1)\, dv_1 
\\ &=&
\frac {n^4}8 \d_n^{2d+2} c_d I^{(2)}_{W}(L) (1+\co(1)) 
\end{eqnarray*}
applying Proposition \ref{prop:Int-rcr} in the last step. This is the analogue to Theorem \ref{th:exp-rcr} in the binomial setting.
The analogue to \eqref{eq:exp-L-ppp-allg} would be
$$ \E_{L, V_n} \rcr (G_L)
=
\frac {n^4}8 \d_n^{2d+2} c_d \int\limits_{\cL} I^{(2)}_{W}(L) dL (1+\co(1)). 
$$
We leave further computations to the interested reader.

\section{Conclusion}

Apart from providing precise asymptotics for the crossing numbers of drawings of random 
geometric graphs, the main findings are the positive covariance and the non-vanishing correlation between the stress and the 
crossing number of the drawing of a random geometric graph. 
Of interest would be whether $\cov_{L} (\rcr(G_L), \st (G, G_L) ) > 0$ for arbitrary graphs $G$. 
There are simple examples of graphs $G$ where this is wrong.
However, we could ask in a slightly weaker form whether at least 
$ \E_V \cov_{L} (\rcr(G_L), \st (G_L) ) >0 ,$
but we have not been able to prove that.


\newpage

\bibliography{random-rcr-numbers}



\end{document}